  \documentclass[conference]{IEEEtran}

   \ifCLASSINFOpdf   
\else
\fi
\usepackage{times,epsfig,latexsym,amssymb,psfrag}
\usepackage{paralist,citesort}
\usepackage{booktabs}
\usepackage{subcaption}
\usepackage{eurosym}
\usepackage{mathrsfs}
\usepackage{enumitem}
\usepackage{graphicx}
\usepackage{epstopdf} 
\usepackage{amsmath}
\usepackage{mathtools}
 
 \newtheorem{theorem}{Theorem}[section]

 \newtheorem{corollary}[theorem]{Corollary}
 
 \newenvironment{proof}[1][Proof]{\begin{trivlist}
 		\item[\hskip \labelsep {\bfseries #1}]}{\end{trivlist}}

 \newenvironment{remark}[1][Remark]{\begin{trivlist}
 		\item[\hskip \labelsep {\bfseries #1}]}{\end{trivlist}}
 
 \newcommand{\qed}{\nobreak \ifvmode \relax \else
 	\ifdim\lastskip<1.5em \hskip-\lastskip
 	\hskip1.5em plus0em minus0.5em \fi \nobreak
 	\vrule height0.75em width0.5em depth0.25em\fi}

\DeclareMathOperator*{\maxi}{maximize}
\begin{document}
 \title{Performance Evaluation and Optimization of LPWA IoT Networks: A Stochastic Geometry Approach}

\author{Amin Azari and Cicek Cavdar\\
KTH Royal Institue of Technology, Email: \{aazari, cavdar\}@kth.se} 
 \maketitle

\begin{abstract}
Leveraging grant-free radio access for enabling low-power wide-area (LPWA) Internet of Things (IoT) connectivity has attracted lots of attention in recent years. Regarding lack of research on LPWA IoT networks, this work is devoted to reliability modeling, battery-lifetime analysis, and operation-control of such networks. We derive the interplay amongst density of the access points, communication bandwidth, volume of traffic from heterogeneous sources, and quality of service (QoS) in communications. The presented analytical framework comprises modeling of interference from heterogeneous sources with correlated deployment locations and time-frequency asynchronous radio-resource usage patterns. The derived expressions represent the operation regions and   rates in which, energy and cost resources of devices and the access network, respectively, could be traded to achieve a given level of QoS in communications. For example, our expressions indicate the expected increase in QoS by increasing number of transmitted replicas, transmit power, density of the access points, and communication bandwidth.    Our results further shed light on scalability of such networks and figure out the bounds up to which, scaling resources can compensate the increase in traffic volume and QoS demand. Finally, we present an energy-optimized operation control policy for IoT devices. The simulation results confirm tightness of the derived analytical expressions, and indicate usefulness of them in planning and operation control of IoT networks.

\end{abstract}
\begin{IEEEkeywords}
5G, Coexistence, Grant-free, Reliability and durability, LPWA IoT.
\end{IEEEkeywords}
 
\IEEEpeerreviewmaketitle
 
 
 \section{Introduction}
Providing  connectivity for  massive Internet-of-Things (IoT) devices is a key driver of 5G \cite{5g_iot}. Until now, several solutions have been proposed  for enabling large-scale  IoT connectivity, including evolutionary and revolutionary solutions \cite{mag_all}.  Evolutionary solutions aim at enhancing connectivity procedure of existing LTE networks, e.g.  access reservation and scheduling improvement \cite{isl,nL}. On the other hand, revolutionary solutions  aim at providing   scalable low-power IoT connectivity by redesigning the access network. In 3GPP LTE Rel. 13, narrowband IoT (NB-IoT) has been announced as a revolutionary solution which handles communications over a 200 KHz bandwidth \cite{ciot}. This narrow bandwidth brings high link budget, and offers extended coverage \cite{ciot}.  To provide autonomous low-latency access to radio resources, grant-free radio access is  a study item in 3GPP IoT working groups, and it is expected to be included in future 3GPP standards \cite{gf31}. 
Thanks to the simplified connectivity procedure, and removing the need for pairing and fine synchronization, grant-free radio access has attracted lots of interests in recent years  for providing  low-power ultra-durable IoT connectivity, especially when more than 10 years lifetime is required.   SigFox and LoRa are two dominant grant-free radio access solutions over the public ISM-band, which is used for industrial, scientific, and medical purposes  \cite{mag_all}.   While energy consumptions of LoRa and SigFox solutions are extremely low, and their provided link budget is enough to penetrate to most indoor areas, e.g. LoRa signal can be decoded when it is  20 dB less than the noise level,  reliability of their communications in coexistence scenarios is questionable \cite{int2,mey}. \cite{int2} presents experimental measurements in such coexistence scenarios, where multiple IoT technologies  are sharing a set of radio resources, and confirms significant impact of interference on IoT communications. Regarding the growing interest in grant-free radio access for IoT communications in public and proprietary cellular networks \cite{gf31,mag_all}, it is required to investigate the reliability, battery lifetime, and scalability of such networks in serving multi-type IoT devices.
\subsection{Literature Study}
Non-orthogonal radio access  has attracted lots of attentions  in recent years as a  complementary radio access scheme for future generations of wireless networks \cite{noma,jsacS}. In literature, non-orthogonal access has been mainly employed in order to increase the network throughput \cite{reem},  reliability \cite{url}, battery lifetime \cite{gf}, and reduce access delay \cite{reem} in serving non-IoT traffic.   
In \cite{miao2016MAC}, grant-free  access to uplink radio resources of cellular networks has been analyzed for intra-group communications of IoT devices. In \cite{gf}, a  novel receiver for grant-free radio access IoT networks has been designed, which benefits from oscillator imperfection
of cheap IoT devices for contention resolution.  In \cite{2d},   outage probability in grant-free access has been studied by assuming a constant received power from all contending devices, which is not the case in practice  regarding the limited transmit-power of IoT devices, as well as lack of channel state information at the  device-side for power control.  The success probability in grant-free radio access  has been also analyzed in \cite{sic,mey} by assuming a   Poisson point process (PPP) distribution of IoT devices. 
 
One sees the research on grant-free radio access has been mainly focused on success probability analysis in homogeneous scenarios, and there is lack of research on performance analysis of large-scale IoT networks with multi-type IoT devices with heterogeneous  communications characteristics. Furthermore, when it comes to the distribution of devices in wide-area IoT networks, PPP has been mainly used. However, this assumption may lead to inaccurate results \cite{math,pcp} due to the cell ranges  that can go up to tens of kilometers \cite{mag_all} and 
 hot-spots.  In hot-spots, e.g. buildings and shopping centers,    a high density of IoT devices exist; while outside them,  a low density of devices exists. Then,  a Poisson cluster process (PCP), which takes the correlation between locations of devices into account,  suits well for the distribution process of  devices in LPWA IoT networks \cite{math,pcp}.  
  
\subsection{Contributions}
  Here, we address an important  problem, not tackled previously: network design  in coexistence scenarios  with grant-free radio access. Enabling IoT connectivity requires deployment of access points (APs)  and allocation of frequency resources, which increase  the network costs. On the other hand, the experienced delay, consumed energy, and success of IoT applications have strong couplings with reliability of data transfer, which is a function of provisioned network resources. This tradeoff is investigated in this work. The main contributions of this work include:
    \begin{itemize}
    \item
Provide a rigorous analytical model of reliability for heterogeneous LPWA IoT networks   in terms of provisioned  resources, e.g. density of the APs, and characteristics of traffic, e.g.  activity factor of each traffic type.
\item
    Provide  an analytical model of battery  lifetime for IoT devices in terms of device's parameters, e.g. battery capacity, and network parameters, e.g. reliability of communications. 
     \item     
     Analyze the tradeoffs among network cost, battery lifetime, and reliability of communications.  Present the operation regions in which tuning a communication parameter, e.g. number of replica transmissions, increases both reliability and battery lifetime, offers a tradeoff between them, and decreases both of them.
     \item
Propose a reliability-constrained lifetime-optimized operation control policy for IoT devices.   
\item
Analyze scalability of the network. Figuring out the bounds up to which, scaling network's and devices' resources can compensate the increase in traffic volume and QoS demand. 
\end{itemize}

The remainder of paper has been organized as follows. System model and problem description are presented in the next section. Modeling of KPIs is presented in section III. Section IV presents the optimized  operation control strategies. Simulation results are presented in section V. Concluding remarks are given in section VI.

\section{System Model and Problem Description}
\subsection{System Model}\label{sys}
 A set of  IoT devices, denoted by $\Phi$, have been distributed according to different spatial PCPs in a wide service area. $\Phi$ comprises of $K$ subsets, $\Phi_k$ for $k\in \mathcal K \buildrel \Delta \over = \{1,\cdots, K\}$, where each subset refers to a specific type of IoT service. Traffic from different subsets differ in the way they use the time-frequency resources, i.e. in frequency of packet generation $1/T_k$, signal bandwidth $w_k$, packet transmission time $\tau_k$, number of replicas\footnote{Practical motivations for modeling such replicas can be found in state of the art IoT technologies like NB-IoT and SigFox in which, coverage extension and resilience to interference  are achieved by repetitions of transmitted packets \cite{ciot,mag_all}. When it is not the case, $n_k=1$ can be used.} transmitted per packet $n_k$, and transmit power $P_k$. Subscript $k$ refers to  the type of IoT devices.  For PCP of type-$k$ IoT traffic, the $(\lambda_k, \upsilon_k, \text{f}({\bf x}))$ tuple characterizes the distribution process in which, $\lambda_k$ is the density of the parent points and $\upsilon_k$ is the average number of daughter points per parent point\footnote{In PCP deployment, we have clusters of devices, where each cluster models a hot-spot. $\lambda_k$ represents density of such clusters of devices, i.e. the parent points. $\upsilon_k$ represents the average  number of devices in each cluster, i.e. the daughter points. Finally, $\text{f}(x)$ represents how devices are distributed in each cluster.}, as defined in \cite{pcp}. Also, $\text{f}({\bf x})$ is an isotropic function representing  scattering density of the daughter points around a parent point, e.g. a normal distribution:
 \begin{equation}\label{nor}\text{f}({\bf x})={\exp(-||{\bf x}-{\bf x}_0||^2/(2\sigma^2))}/{\sqrt{2\pi\sigma^2}},\end{equation}
 where $\sigma$ is the variance of distribution and ${\bf x}_0$ is the location of parent point.   A frequency spectrum of $W$ is shared for  communications, on which the power spectral density of noise is denoted by $\mathcal N$. We aim at collecting data from a subset of $\mathcal K$, denoted by $\phi$, where $|\phi|\le |\mathcal K|$.   Devices in $\phi$  may also share a set of semi-orthogonal codes denoted by $\varpi $ with cardinality $|\varpi|$, which reduces the interference from other devices reusing the same radio resource with a different code by factor of $\mathcal Q$. Examples of such codes are semi-orthogonal spreading codes in LoRa
technology \cite{mag_all}.

%
%

\section{Analytical Modeling of KPIs}
\subsection{Modeling of Reliability}
In the grant-free radio access system,  transmitting devices are asynchronous in time and frequency domains, and hence, the received packets at the receiver could have partial overlaps in time-frequency.  To model reliability in communications,  we first derive   analytical  models for interference in subsection \ref{siI}, and for probability of success in subsection \ref{su1}. These models are then employed in deriving reliability of communications in subsection \ref{rels}.

\subsubsection{Interference Analysis}\label{siI}
We assume a type-$i$ device has been located at point $\bf z$ in a 2D plane, and its respective AP has been located at the origin.  In order to derive probability of success in data transmission from the device to the AP, we need to characterize the received interfere at the AP.   A common practice in interference analysis is to determine its moments, which is possible by finding its generating function, i.e. the Laplace functional \cite{adhoc,math}.  Towards this end, let us introduce three stationary and isotropic processes: i) $ \Psi^{(1)} =\cup_{k\in \mathcal K} \Psi_k^{(1)}$,  where $\Psi_{k}^{(1)}$ represents the PCP containing locations of type-$k$ transmitting nodes which are reusing radio resources with a similar code to the code\footnote{Note: as mentioned in the system model, devices in $\phi$ share a set of semi-orthogonal codes for partial interference management.} of transmitter of interest; ii) $ \Psi^{(2)} =\cup_{k\in \mathcal K}  \Psi_k^{(2)}$,  where $ \Psi_{k}^{(2)}$ represents the PCP containing locations of type-$k$ transmitting nodes which are reusing radio resources with a different code (or no code, in case $k\notin \phi$) than the transmitter of interest; and iii) $\Psi=\cup_{j\in\{1,2\}} \Psi_k^{(j)}$.
   For an AP located at the origin, the Laplace functional of the received interference at the receiver is given by:
 \begin{align}
 \mathcal L_{I_{\Psi}}(s)&=\mathbb E\big[\exp(-s I_{\Psi})\big]\label{base}\\
 &=\mathbb E\big[  \prod\nolimits_{j\in \{1,2\}} \prod\nolimits_{k\in \mathcal K}  \prod\nolimits_{{\bf x}\in \Psi_k^{(j)}}\mathcal L_h({sQ_j P_k  \text{g}({\bf x})})\big],\nonumber
  \end{align}
where $ Q_j P_k  \text{g}({\bf x})$ is the average received power  due to a type-$k$ transmitter  at point ${\bf x}$, $Q_1=1$, $Q_2=\mathcal Q$, and $\mathcal Q$ is the rate of rejection of interference between two devices with different multiple access codes, as defined in section \ref{sys}. Also,  $h$ is the power fading coefficient associated with the channel between the device and the AP, and $\mathcal L_{h}\big(s Q_j  P_k \text{g}({\bf x})\big)$ is the Laplace functional of the received power. We consider the following general path-loss model 
$ \text{g}({\bf x}) = 1/(\alpha_1 + \alpha_2||{\bf x}||^\delta),$  where $\delta$ is the pathloss exponent, and $\alpha_1$ and $\alpha_2$ are control parameters.  When $h$ follows Nakagami-$m$ fading, with the shaping and spread parameters of $m\in \mathbb Z^i$ and $\Omega>0$ respectively, the probability density function (PDF) of the power fading coefficient is given by:
 \begin{equation}\text{p}_h(q)= \frac{1}{{\Gamma(m)}}(\frac{m}{\Omega})^m q^{m-1}\exp\big({-\frac{mq}{\Omega}}\big),\label{nm}\end{equation}
 where $\Gamma$ is the Gamma function.
 Then, using Laplace  table, $\mathcal L_{h}\big(sQ_j  P_k \text{g}({\bf x})\big)$   is derived as:
\begin{equation}\label{laph}L_h(s  Q_j P_k \text{g}({\bf x}))={\big(1+{\Omega}s P_k \text{g}({\bf x})/m\big)^{-m}}.\end{equation} 
By inserting \eqref{laph} in \eqref{base} and considering the fact that the received interferences from different devices are independent, we have: 
 \begin{equation}
 \mathcal L_{I_{\Psi}}(s)=  \prod\nolimits_{j,k}\mathbb E_{{\bf x},{\bf y}}\bigg[  \prod\nolimits_{{\bf y}\in \Theta_{k}}\big(  \prod\nolimits_{{\bf x}\in \theta_{\bf y}^{(j)}}u({\bf x},{\bf y})\big)\bigg],\nonumber
  \end{equation}
  where $k\in\mathcal K$, $j\in\{1,2\}$,   the set of parent points of type-$k$ is denoted by $\Theta_{k}$,  and transmitting nodes which are daughter points of $y$ as $\theta_{\bf y}^{(j)}$.
  Also, $\mathbb E_x$ represents expectation over $x$,  
and  $$u({\bf x},{\bf y})={\big(1 {+}{\Omega}s Q_j P_k \text{g}({\bf x} {-}{\bf y})/{m}\big)^{-m}}.$$ 
The received interference over the packet of interest can be decomposed into two parts: i) interference from transmitters belonging to the cluster of transmitter, i.e.  daughter points of the same parent point; and ii)  other  transmitters.
Let us denote the Laplace functional of interference from the former and latter transmitters as $\mathcal L_{I_{\Psi}}^{\dag}(s)$ and $\mathcal L_{I_{\Psi}}^{\ddag}(s)$ respectively. Then, we have: \begin{equation}\label{taj}\mathcal L_{I_{\Psi}}(s)=\mathcal L_{I_{\Psi}}^{\ddag}(s)\mathcal L_{I_{\Psi}}^{\dag}(s).\end{equation}
Using equation (18)  in \cite{math}, and  by conditioning on $\Theta_{k}$ and $\theta_{\bf y}^{(j)}$,  one has:
 \begin{align}
 &\mathcal L_{I_{\Psi}}^{\ddag}(s)\label{li}\\
&  =  \prod\nolimits_{j,k}\mathbb E_{y}\bigg[  \prod\limits_{{\bf y}\in \Theta_{k}}\big\{\exp\big(\text{-}\hat \upsilon_{k,j}  \int\nolimits_{\mathbb R^2} [1\text{-}{u({\bf x},{\bf y})}]\text{f}({\bf x})d {\bf x}\big)\big\} \bigg]\nonumber,\\
 &  =\exp\big( \text{-} {\textstyle \sum\limits_{j,k}}\lambda_k \int\limits_{\mathbb R^2}\big\{1\text{-}\exp\big(\text{-}\hat\upsilon_{k,j} \int\limits_{\mathbb R^2}[1\text{-} {u({\bf x},{\bf y})}]\text{f}({\bf x})d {\bf x} \big)\big\}d {\bf y}\big).\nonumber
 \end{align}
  Also, in \eqref{li} the average numbers of interfering type-$k$ devices in each cluster for $j\in\{1,2\}$ are denoted as  $\hat \upsilon_{k,1}=\upsilon_k\frac{n_k\tau_k}{T_k} \frac{w_k}{W}\frac{1}{|\varpi| }$  and $\hat \upsilon_{k,2}=\upsilon_k\frac{n_k\tau_k}{T_k} \frac{w_k}{W}\frac{|\varpi| -1}{|\varpi| }$ for $k\in\phi$. In these two expressions, the first fraction represents the percentage of  time in which device is active, i.e. the time activity factor, the second fraction represents the ratio of bandwidth that device occupies in each transmission, i.e. the frequency activity-factor, and the third fraction represents the code-domain activity factor, i.e. the probability that two devices select the same code, i.e. $\frac{1}{|\varpi| }$, or different codes $\frac{|\varpi| -1}{|\varpi| }$. Then, for $k\notin \phi$, in which devices don't share semi-orthogonal codes for communications, it is clear that $\hat \upsilon_{k,1}=0$, and $\hat \upsilon_{k,2}=\upsilon_k\frac{n_k\tau_k}{T_k} \frac{w_k}{W}$.
Following the same procedure   used for deriving $\mathcal L_{I_\Psi}^{\ddag}(s)$, one can derive   $\mathcal L_{I_\Psi}^{\dag}(s)$  as:
\begin{align}
\mathcal L_{I_\Psi}^{\dag}(s)\text{=}&\prod\nolimits_{j\in\{1,2\}}\mathbb E_y \big[\mathbb E_x [  \prod\nolimits_{{\bf x}\in \theta_{\bf y}^{(j)}}u({\bf x},{\bf y}) ]\big]\label{lik}\\
\text{=}& \int\nolimits_{\mathbb R^2}\exp\big(\text{-}{\textstyle\sum_j}\hat \upsilon_{i,j}  \int\nolimits_{\mathbb R^2} \big(1\text{-}{u({\bf x},{\bf y})}\big)\text{f}({\bf x})d {\bf x}\bigg)\text{f}({\bf y}) d {\bf y}\nonumber.\end{align}

\subsubsection{Probability of Success in Transmission}\label{su1}
Let $N$ denote the additive
noise at the receiver. Using the interference model,  probability of success in  packet transmission of  a type-$i$ device, located at $\bf z$, to the AP, located at the origin, is: 
\begin{align}
\text{p}_{\text{s}}(i,{\bf z})&=\text{Pr}({  P_ih \text{g}({\bf z})}\ge[{N+I_{\Psi}}] \gamma_{\text{th}})\label{suc}\\
&\buildrel (\text{a}) \over =   \sum\limits_{\nu=0}^{m\text{-}1}\frac{1}{{\nu}!}\int\nolimits_{0}^{\infty}\exp({-}\frac{\gamma_{\text{th}}m q}{\Omega   P_i \text{g}({\bf z})}) q^{\nu}  d \text{Pr}(I_\Psi\text{+}N\ge q)\nonumber\\
&\buildrel (\text{b}) \over =  \sum\nolimits_{{\nu}=0}^{m\text{-}1}\frac{(-1)^{\nu}}{{\nu}!}[\mathcal L_{I_{\Psi}}(s)\mathcal L_{N}(s)]^{({\nu})} \big|_{s=\frac{\gamma_{\text{th}}m}{\Omega   P_i \text{g}({\bf z})}}, \nonumber
\end{align}
where $[F(s)]^{({\nu})}=\frac{\partial^{\nu}}{\partial s^{\nu}}F(s)$, (a) follows from \cite[Appendix~C]{adhoc}  and equation \eqref{nm} in which $\text{p}_h(q)$ has been defined, and finally (b) follows from \cite[Lemma~3.1]{alm} and the fact that $\mathcal L (t^n \text{f}(t))=(-1)^n\frac{\partial^n}{\partial s^n}F(s)$.
Furthermore, $L_{I_{\Psi}}$ has been characterized in \eqref{li} and \eqref{lik}, and $\mathcal L_N(s)$ is the Laplace transform of noise. 
In order to get insights on how coexisting services affect each other, in the following we focus on  $m=1$, i.e. the Rayleigh fading, and present a  closed-form approximation of the success probability. In section \ref{simsec}, we will evaluate tightness of this expression.
\begin{theorem}\label{t1}
For $m=1$, success probability  in packet transmission  can be  approximated as:
\begin{align}
&\text{p}_{\text{s}}(i,{\bf z})\approx \text{P}_{{\text{\tiny N} }} \big[\exp\big(-\sum\limits_{j\in\{1,2\}} \sum\limits_{k\in\mathcal K}\lambda_k\hat \upsilon_{k,j} \text{H}({\bf z},1, \frac{Q_j P_k\gamma_{\text{th}}}{\Omega  P_i})\big)\big]\nonumber\\
&\hspace{1cm}\times\exp\big(-\sum\nolimits_{j\in\{1,2\}}{\hat\upsilon_{i,j}}  \text{H}({\bf z},\text{f}^*({\bf x}), \frac{Q_j\gamma_{\text{th}}}{\Omega })\big),\label{ps}
\end{align}
  where  $\text{f}^*(\cdot)=\text{conv}\big(\text{f}(\cdot),\text{f}(\cdot)\big)$, 
  \begin{align}
  \text{H}\big({\bf z},\text{f}^*({\bf x}), \xi)&=\int\nolimits_{x\in\mathbb R^2}\frac{\text{g}({\bf x})}{\text{g}({\bf  x})+\text{g}({\bf z})/\xi}\text{f}^*({\bf x})d {\bf x}\label{hf},\\
   \text{P}_{{\text{\tiny N} }}&=\exp\big(-\mathcal N\gamma_{\text{th}}/[\Omega  P_i \text{g}({\bf z})] \big)\label{pn},
  \end{align}
  and $\mathcal N$ is the noise power.
\end{theorem}
\begin{proof}
The proof is given in theorem 3.2 of the extended version \cite{opd}.
\end{proof}
  $\text{H}({\bf z},\text{f}^*({\bf x}),\xi)$ and $\text{H}({\bf z},1,\xi)$  could be derived in closed-form for most well-known pathloss and distribution functions, as follows.
   \begin{corollary}
For $\text{g}({\bf x})=\alpha||{\bf x}||^{-\delta}$, 
\begin{equation}\label{hfd} \text{H}({\bf z},1,\xi)=  ||{\bf z}||^2 \xi^{\frac{2}{\delta}} 2\pi^{2} \text{csc}({2\pi/\delta})/\delta.\end{equation}
\end{corollary}
\begin{proof}
By change of coordinates, ${\bf x}\to (r,\theta)$, we have:
\begin{align}
\text{H}\big({\bf z},1, \xi)&=\int\nolimits_{x\in\mathbb R^2}\frac{\alpha{||\bf x||}^{-\delta}}{\alpha{||\bf x||}^{-\delta}+\alpha{||\bf z||}^{-\delta}/\xi} d {\bf x},\nonumber\\
&=2\pi\int\nolimits_{0}^{\infty}\frac{1}{1+(r/||{\bf z}||)^\delta/\xi}  { r  d r}.\nonumber
\end{align}
Solving this integral by using \cite[Eq.~3.352]{seri}, \eqref{hfd} is derived.
\qed
\end{proof}

 \begin{corollary}\label{cne}
For $\text{g}({\bf x})=\alpha||{\bf x}||^{-4}$,  and $\text{f}({\bf x})$ given in \eqref{nor},
\begin{align}
 \text{H}({\bf z},\text{f}^*({\bf x}),\xi)=&  \frac{||{\bf z}||^2  }{4 \sigma^2\sqrt\xi    }\bigg[\text{ci}(\frac{ ||{\bf z}||^2 }{4\sigma^2\sqrt{\xi}}  )\sin(\frac{||{\bf z}||^2 }{4\sigma^2\sqrt{\xi}}  )-\nonumber\\
&\hspace{1.7cm}\text{si}(\frac{||{\bf z}||^2 }{4\sigma^2\sqrt{\xi}}  )\cos( \frac{||{\bf z}||^2 }{4\sigma^2\sqrt{\xi}} )\bigg],\nonumber
\end{align}
where $\text{si}(\cdot)$ and $\text{ci}(\cdot)$ are well-known sine and cosine integrals, as follows:
$$\text{si}(x)=-\int\nolimits_{x}^{\infty}\frac{\text{sin} (t)}{t}dt,\hspace{2mm} \text{ci}(x)=-\int\nolimits_{x}^{\infty}\frac{\text{cos}(t)}{t}dt.$$ 
\end{corollary}
\begin{proof}
The proof is given in corollary 3.4 of the extended version \cite{opd}.
\end{proof}

\begin{remark}\label{r1}
Analysis of $\text{H}\big({\bf z},\text{f}^*({\bf x}), \xi)$ shows that it can be well approximated by $1$ for $\frac{\sqrt\xi||{\bf z}||^2   }{4  \sigma^2  }\gg1$. For theorem \ref{t1} in which $\xi=Q_j\gamma_{\text {th}}/{\Omega}$, $ \text{H}\big({\bf z},\text{f}^*({\bf x}), \xi)\approx 0$ for $j=1$ because $Q_1=\mathcal Q\approx 0$; and $\text{H}\big({\bf z},\text{f}^*({\bf x}), \xi)\approx 1$ for $j=2$ when $z\gg z_0 \buildrel \Delta \over =\frac{2\sigma \sqrt[4]\Omega}{\sqrt[4]{\gamma_{\text{th}}}}$ because $Q_2=1$. 
\end{remark}

\begin{remark}
From theorem \ref{t1}, one sees that probability of success, $\text{p}_{\text s}(i,{\bf z})$, is a function of $||\bf z||$ rather than phase of $\bf z$.  Then, hereafter we use  $\text{p}(i,z)$ to denote  probability of success for communication distance of $z$.
\end{remark}

Until now, we have derived the probability of success for a given communication distance to an AP. In the following, we investigate success probability  where multiple APs might be able to decode a packet, i.e. the coverage areas of neighboring APs are overlapping. 
Regarding the fact that  theorem \ref{t1} provides probability of success as a function of communication distance, given the distribution process of APs, the expected communication distance to the neighboring APs, and hence,  probability of success in data transmission   could be derived. In PPP deployment of APs with density $\lambda_{\text{a}}$, the PDF of distance from a random point to the $\ell$th nearest AP, denoted by $d_{\ell}$ is given by \cite{dis}:
$$ \text{P}_{d_{\ell}}(r)=\exp(-\lambda_{\text a}\pi r^2) {2(\lambda_{\text a}\pi r^2)^{\ell}}/[{r({\ell}-1)!}].$$
Then, one can derive the average probability of success in packet transmission from a random point for type-$i$  as:
\begin{equation}\text{P}_\text{s}(i)=  1-\prod\nolimits_{{\ell}=1}^{\ell_{\max}} \int\nolimits_{0}^{\infty} \big(1-\text{p}_\text{s}(i,r)\big) ~ \text{P}_{d_{\ell}}(r) dr.\label{cov}\end{equation}
\begin{theorem}\label{t3}
For $\text{f}(x)$ given in \eqref{nor}, and $\text{g}({\bf z})=\alpha||{\bf z}||^{-4}$,  we have:
 $$\text{P}_\text{s}(i)\approx  1-\prod\nolimits_{{\ell}=1}^{\ell_{\max}} \big[1-\frac{X_0}{\sqrt{{X_1}^{\ell-1}}} \exp(\frac{{X_2}^2}{4{X_1}^2} ) \mathcal G(X_3,\ell)\big],$$
  \begin{align}
&\text{where } X_0=\frac{(\lambda_{\text a}\pi)^\ell}{(\ell-1)!}\exp\big(-{\hat\upsilon_{i,2}} \big), X_1=\frac{ \mathcal N\gamma_{\text{th}}}{\Omega  P_i \alpha},\nonumber\\ 
  & X_2=\sum\limits_{j,k} \lambda_k\hat \upsilon_{k,j} 
  (\frac{\gamma_{\text{th}}Q_j P_k}{\Omega  P_i})^{0.5} \frac{\pi^{2}}{2} \text{csc}(\frac{\pi}{2})+\lambda_{\text a}\pi, X_3=\frac{X_2}{2\sqrt{X_1}}.\nonumber
   \end{align}
 Also,  $\mathcal G(X_3,\ell)=\int\nolimits_{\frac{{X_2}^2}{2X_1}}^{\infty}  (z\text{-}X_3)^{(\ell-1)}\exp(-z^2)dz,$ and could be derived for any $\ell$ in the form of error function, e.g. for $\ell_{\max}=2$: 
 \begin{align}
&\mathcal G(X_3,1)=-(\sqrt{\pi}(\text{erf}(X_3) - 1))/2,\nonumber\\
&\mathcal G(X_3,2)=\exp(-X_3^2)/2 + (X_3 \sqrt{\pi}(\text{erf}(X_3) - 1))/2.\nonumber
\end{align}
\end{theorem}
\begin{proof}
The proof is given in theorem 3.5 of the extended version \cite{opd}.
\end{proof}
 \subsubsection{Reliability of IoT Communication}\label{rels}
Now, we have the required tools to investigate reliability of IoT communications. Once a type-$i$ device has a packet to transmit, it transmits $n_k$ replicas of the packet, and listens for ACK from the AP(s). If No ACK is received in a bounded listening window, device retransmits the packet, and this procedure could be repeated up to $B_{i}-1$ times, where the bound may come from the fair use of the shared medium \cite{int2,mag_all} or expiration of data.   If data transmission is unsuccessful in $B_i$ attempts, we call it an outage event.  The probability of outage for type $i$ in  such setting could be derived as:
\begin{equation}\label{rel}\text{P}_{\text o}(i)=\big[1-\text{P}_{\text s}(i) \big]^{n_iB_i},\end{equation}
where $\text{P}_{\text s}(i)$ has been derived in theorem \ref{t3}.


 \subsection{Battery Lifetime Performance (Durability)}\label{bl}
Packet generation at each device for most reporting IoT applications can be seen as a Poisson process \cite{3g}. Then, one can  model energy consumption of a device as a semi-regenerative process where the regeneration point  has been located  at the end of each successful data transmission epoch \cite{nL}. For a given device of type-$i$, let us denote  the stored energy in batteries as $E_{0}$, static energy consumption per reporting period for data acquisition from environment and processing as $E_{\text{st}}$, circuit power consumption in transmission mode as $P_{c}$, and inverse of power amplifier efficiency as $\eta$. Then, the expected battery lifetime  is \cite{nL}: 
\begin{equation} 
\mathbb L(i)= \frac{E_{0}}{{E_{\text{st}}+\hat \beta_i E_\text{c}+   \hat\beta_in_i (\eta P_{i}+ P_{\text c}) \tau_i}}T_i,\label{lif}
\end{equation}
where $E_\text{c}$  represents the average energy consumption in listening after each trial  for ACK reception, and $\hat \beta_i$ represents the average number of  trials and is derived as:
\begin{equation}\label{beta}\hat \beta_i=\sum\nolimits_{j=1}^{B_i}j\big[1\text{-}[1\text{-}\text{P}_{\text{s}}(i )]^{n_i}\big]\big[1\text{-}\text{P}_{\text s}(i )\big]^{n_i[j-1]},\end{equation}
where $\text{P}_{\text s}(i)$ have been derived in  theorem \ref{t3}.

\section{Optimized Operation Control}
From the battery lifetime analysis in \eqref{lif}, one sees that battery lifetime of devices may decrease in $n_i$ and $P_i$ because of the potential increase in the energy consumption per reporting period.   Furthermore, when reliability of communication is lower than a threshold, increase in $n_i$ and $ {P}_i$ may decrease the need for listening to the channel for ACK arrival and retransmissions, and hence, increasing $n_i$ and $ {P}_i$ may increase the battery lifetime. Taking this into account, one sees there should be an operation point beyond which, increase in $ {P}_i$ and/or $n_i$ offers a tradeoff between reliability and lifetime, and before it,  increase in $ {P}_i$ and/or $n_i$ increases both reliability and durability of communications. This observation will be evaluated using simulation results in the next section. Here, we aim at finding the optimized operation point of the network with respect to the battery lifetime.  Using the  battery lifetime definition in \eqref{lif}, one may define the optimization problem for deriving the optimized operating point of type $i$ IoT devices  as follows:

\begin{align}
\maxi_{n_i, P_i } & \hspace{3mm}\mathbb L(i); \label{op2}\\
&\text{s.t.:} ~\text{P}_{\text o}(i)\le \text{P}_{\text o}^{\text{req}}(i), n_i\le n_{max}, P_i\le P_{\max}, \nonumber
\end{align}
where $\text{P}_{\text o}^{\text{req}}(i)$ is the maximum tolerated outage probability for type $i$ IoT devices.
The reliability constraint in \eqref{op2} could be rewritten as the minimum required success probability in communications as follows:
\begin{equation}\label{con}1-\sqrt[n_iB_i]{\text{P}_{\text o}^{\text{req}}}(i)\le \text{P}_{\text s}(i).\end{equation}
Furthermore,   by using the $\text{P}_{\text s}(i)$ expression in theorem \ref{t3},  we have:
 \begin{align}
\text{P}_{\text s}(i)&= \int\nolimits_{0}^{\infty} X_0 \exp(\text{-}X_5r^{2})2rdr\nonumber\\
&=\frac{0.5\sqrt{\pi}{\lambda_{\text a}\pi}\exp\big(-{\hat\upsilon_{i,2}} \big)}{\sum_{k}\lambda_{k}\hat \upsilon_{k,2} 
  (\frac{P_k\gamma_{\text{th}}}{P_i\Omega   })^{0.5} \frac{\pi^{2}}{2} \text{csc}(\frac{\pi}{2})\text{+}\lambda_{\text a}\pi\text{+}\frac{ \mathcal N\gamma_{\text{th}}}{\Omega  P_i \alpha}},\label{den}
\end{align}
in which, $\ell_{\max}=1$, $\delta=2$, and  $\mathcal Q\approx1$ have been assumed for brevity of expressions. Also,   $X_5$ is an auxiliary variable equal to the denominator of \eqref{den}.
The expression in \eqref{den} could be rewritten as:
\begin{equation}\text{P}_{\text{s}}(i)=\frac{D_0}{\frac{1}{\sqrt{P_i}}D_1+\lambda_{\text a}\pi+\frac{ \mathcal N\gamma_{\text{th}}}{P_i\Omega   \alpha}},\label{rpi}\end{equation}
where the auxiliary variables $D_0$ and $D_1$ are defined as:
\begin{align}
D_0&=0.5\sqrt{\pi}{\lambda_{\text a}\pi}\exp\big(-{\hat\upsilon_{i,2}} \big),\nonumber\\
 D_1&=\sum\nolimits_{k}\lambda_{k}\hat \upsilon_{k,2} 
  (\frac{P_k\gamma_{\text{th}}}{\Omega   })^{0.5} \frac{\pi^{2}}{2} \text{csc}(\frac{\pi}{2}).\nonumber
  \end{align}
Satisfying  \eqref{con} with equality, we have:
 $$\sqrt[n_iB_i]{\text{P}_{\text o}^{\text{req}}}(i)=1- \frac{D_0}{\frac{1}{\sqrt{P_i}}D_1+\lambda_{\text a}\pi+\frac{ \mathcal N\gamma_{\text{th}}}{P_i\Omega   \alpha}}.$$
 By simplifying the expression, $n_i$ is derived as a function of $B_i$ as follows:
\begin{equation} {n_i}=\left \lceil {\log(\sqrt[B_i]{\text{P}_{\text o}^{\text{req}}})}\bigg/{\log(1- \frac{D_0}{\frac{1}{\sqrt{P_i}}D_1+\lambda_{\text a}\pi+\frac{ \mathcal N\gamma_{\text{th}}}{P_i\Omega   \alpha}})}\right \rceil.\label{rni}\end{equation}
Also, the constraint on $n_i$ is translated to a constraint on $P_i$ as:
$$P_i\ge P_{\min} \buildrel \Delta \over =\big(\frac{-{D_1}\text{+}\sqrt{{D_1}^2\text{-}4\frac{\mathcal N\gamma_{\text{th}}}{\Omega \pi}(\lambda_{\text a}\pi\text{-}\frac{D_0}{1\text{-}\sqrt[n_{\max}B_i]{\text{P}_{\text o}^{\text{req}}}})}}{2(\lambda_{\text a}\pi\text{-}\frac{D_0}{1-\sqrt[n_{\max}B_i]{\text{P}_{\text o}^{\text {req}}}})}\bigg)^2.$$
 Then, the optimization problem in \eqref{op2} reduces to a simple search over $ P_{\min} \le \mathcal{P}_i\le P_{\max} $ for minimization of :
 \begin{equation} \label{ecprp}
 {\hat \beta_i E_\text{c}+   \hat \beta_in_i (\eta P_{i}+ P_{\text c}) \tau_i},
\end{equation}
in which $n_i$ has  been found as a function of $P_i$ in \eqref{rni},  $\hat \beta_i$ has been found as a function of $\text{P}_{\text s} (i)$ and $n_i$ in \eqref{beta}, and $\text{P}_{\text s} (i)$ has been found as a function of $P_i$ in \eqref{rpi}. This operation control optimization problem is  investigated numerically in the next section (Fig. \ref{oo}).


\begin{table}[t!]
\centering \caption{Simulation Parameters   }\label{sim}
\begin{tabular}{p{3.4 cm}p{4.6 cm}}\\
\toprule[0.5mm]
{\it Parameters }&{\it Value}\\
\midrule[0.5mm]
Service area &  $20\times 20 \text{ Km}^2$\\
Pathloss & $133+38.3\log (\frac{x}{1000})$\\
Thermal noise power & $-174$ dBm/Hz \\
Distribution of devices & PCP$\big(\lambda_i\times$1e-6,$200$, Eq. \eqref{nor} with $\sigma$=100)\\ 
Packet arrival of each device & Poisson distributed with average reporting period ($T_i$) of 300 s\\
Packet transmission time ($\tau_i$) & 100 ms\\
Signal BW& 10 KHz\\
$E_0,P_{\text c} , E_{\text{st}}=0.5 E_{\text c}$& 1000 J, 10 mW,    0.1 J\\
$P_{\text r}$, $P_{\text a}$ & 0.5 W, 1.5 W\\
$\gamma_{\text{th}}$, $|\varpi|$, $\eta$&1,1,0.5\\
$P_i, n_i,\lambda_{a}, W$ &  Default: 21 dBm, 1, 5.5{\rm e}-8, 100 KHz  \\
$\ell_{\max},\mathcal Q$&1, 0\\
\bottomrule[0.5mm]
\end{tabular}
\end{table}
 \section{Performance Evaluation}\label{simsec}
In order to investigate  usefulness of our findings in 
IoT-network planning and operation control, here  we implement a MATLAB simulator for a heterogeneous IoT network.  In our simulator, 2 types of IoT devices have been considered, that differ in the distribution processes describing  locations of their respective nodes, and communications' parameters such as transmit power. Motivations for this setup are  the coexistence of  IoT technologies over the public ISM spectrum, e.g. SigFox and LoRa \cite{int2}, and the coexistence of different IoT services over cellular networks, which are sharing a set of uplink resources, as described in \cite{gf31}. For type $i$, the distribution process of locations  is characterized by PCP$(\lambda_i,\upsilon_i,f(x))$, where $\lambda_i$ is the density of cluster points (in Km$^{-2}$), $\upsilon_i=200$ is the average number of nodes in each cluster, and distribution of cluster nodes around the cluster center, i.e. $f(x)$, is modeled by a normal distribution with standard deviation of 100 meters. The reliability constraint is described as $\text{p}_{\text s}(i,d_{\text{eg}})$, where $d_{\text{eg}}=\sqrt{{1}/{(\pi\lambda_{\text a})}}$ is equivalent to the cell-edge communication distance in the case of grid deployment of APs. The packet arrival at each node follows a PPP with rate $\frac{1}{T_i}$.   The default values of other parameters can be found in Table \ref{sim}.

 First, we investigate tightness of the derived analytical expressions. By considering an IoT network comprising of two IoT types with different distributions and transmit powers, Fig. \ref{val} represents probability of success in packet transmission  for type-1 as a function of distance from the AP. One sees that the analytical model matches well with the simulation results. We have further depicted the contributions of noise, interferences from the same and other clusters of type-1 devices, as well as interference from type-2 devices. Regarding the fact that transmit power of type-2 devices is 4 dB higher than type-1 devices in this figure, it is clear that  interference from type-2 traffic (plus-marked curve) is the most limiting factor. 
 
 \begin{figure}[t!]
 	\centering
 	\includegraphics[width=3.5in]{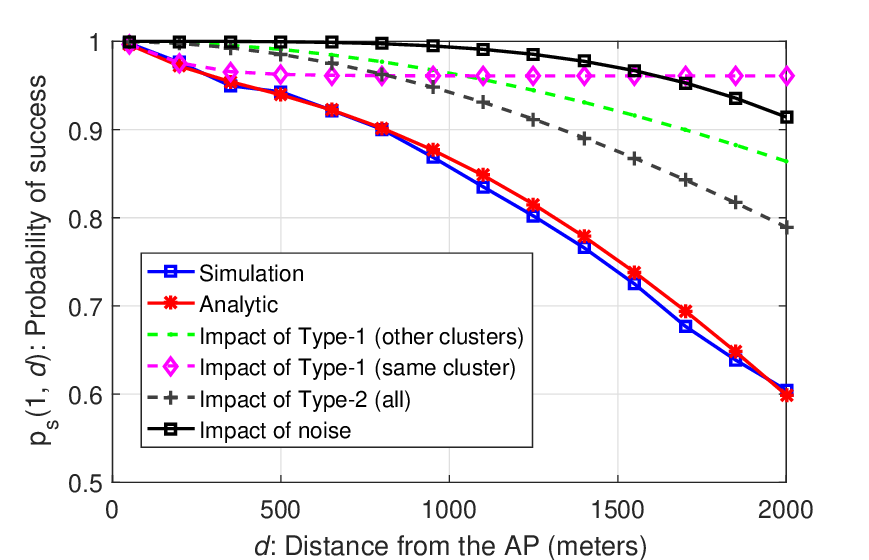}
 	\caption{Validation of analytical and simulation results. Device distribution:  $K$=$2$, $\lambda_1$=0.19, $ \lambda_2$=3.8, $\upsilon_1$=1200, $\upsilon_2$=30, $P_1$=21 dBm, and $P_2$=25 dBm. }
 	\label{val}
 \end{figure}

\begin{figure}[t!]
    \centering
    \begin{subfigure}[t]{0.5\textwidth}
   \centering
  \includegraphics[width=3.5in]{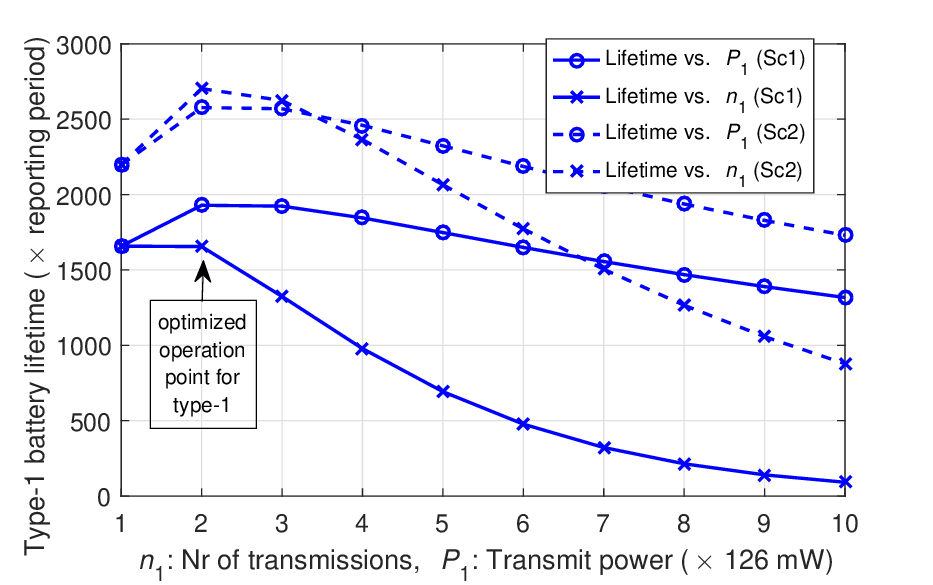}
\caption{Battery lifetime for type-1}\label{oo2}
    \end{subfigure}%
\\
    \begin{subfigure}[t]{0.5\textwidth}
   \centering
    \includegraphics[trim={0.35cm 0.00cm 0cm 0cm},clip,width=3.5in]{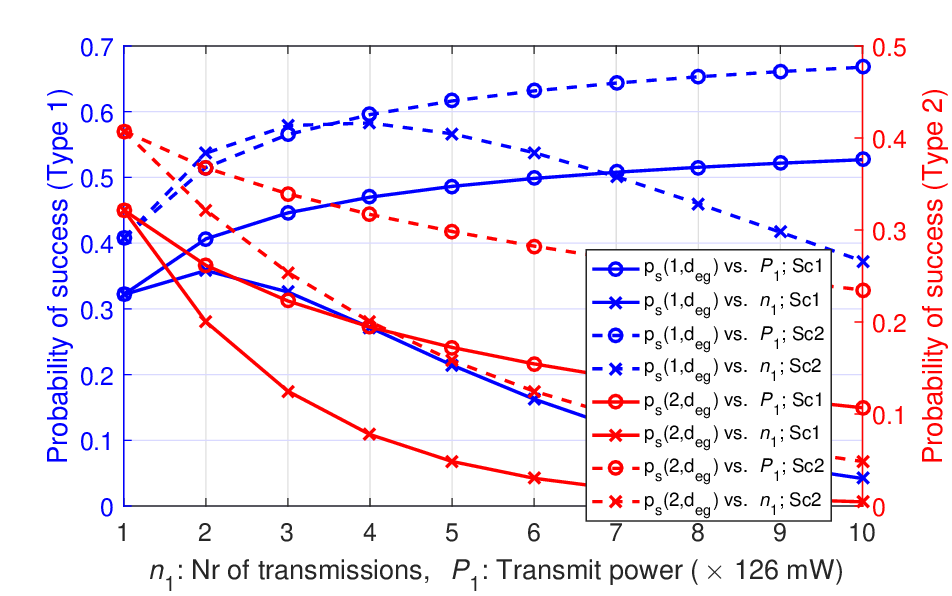}
\caption{Probability of success in transmission for  type-1 and type-2 devices}\label{oo1}
    \end{subfigure}%
\caption{Optimized operation control ($K=2, \lambda_2$=2.4, $\lambda_1$=2.4 in Sc1 and $\lambda_1$=1.2 in Sc2).   In circle-marked curves, $n_1=1$ and $P_1$ is varying. In  plus-marked curves, $P_1=126$ mW and $n_1$ is varying.   }  \label{oo}
\end{figure}

 Fig. \ref{oo} represents the  interplay among success probability, battery lifetime, $n_i$, and $P_i$. The $x$-axis in Fig. \ref{oo2} and Fig. \ref{oo1} represents  $P_1$ for circle-marked curves, and $n_1$ for cross-marked curves. In these figures, Sc1 and Sc2 differ in density of type-2 devices, which is 2.4 in Sc1, and 1.2 in Sc2.  One observes in Fig. \ref{oo2} that battery lifetime is a quasi-concave function of both $P_i$ and $n_i$. Furthermore, in Sc1, where density of nodes is higher than Sc2, battery lifetime decreases significantly by increase in the number of replica transmissions. In both scenarios, we see that  the energy-optimized operation strategy for type-1 devices is to send 2 replicas per data packet to maximize their battery lifetimes. Fig. \ref{oo1} represents the success probability for type-1 and type-2 traffic as a function of $n_1$ and $P_1$. One sees that success probability for type-1 increases to a point beyond which, the resulting interference from extra transmitted packets starts deteriorating the performance. On the other hand, increase in the transmit power for type-1 devices,  increases the success probability for this type and severely decreases the performance of type-2 devices. It is also worthy to note that in Fig. \ref{oo1},  success probability increases in $n_1$ till $n_1=4$, however, from the battery lifetime analysis in Fig. \ref{oo2}, it is evident that battery lifetime decreases in $n_1$ for $n_1\ge 3$. To conclude, we see that increase in the number of replica transmissions, i.e. $n_1$, increases both battery lifetime and reliability  for $n_1\in\{1,2\}$, offers a tradeoff between battery lifetime and reliability for $n_1\in\{3,4\}$, and decreases both reliability and battery lifetime for $n\ge 5$. These results confirm   importance of the derived results  in this work, as they shed light to the operation point after which, it is not feasible to trade battery lifetime in hope of reliability.

Scalability analysis has been presented in Fig. \ref{scc}. The analytical model of reliability has been found in \eqref{rel} as a function of: i) transmit power, ii) number of replica transmissions, iii) density of APs, and iv) bandwidth of communications. Fig. \ref{scc} represents the rate at which, the amount of provisioned resources at the network-side, or energy resources at the device-side, could be scaled to comply with the increase in the level of required reliability. It is clear that transmit power of devices could be increased up to a certain level in order to combat  noise. However, beyond a certain point, increase in the transmit power cannot increase the success probability because it cannot compensate the impact of interference. On the other hand, one sees that increase in the number of replicas per packet could be leveraged to increase reliability of communications. However, there is a saturation point  in scenarios with higher densities of nodes, where increasing number of replicas increases traffic load significantly, and may even reduce reliability of communications. Example of such event was observed in Fig. \ref{oo1} for $n_1\ge5$. Finally, the rate of increase in reliability of communications by  increasing the number of APs, which  reduces the communications' distance, and increasing the  bandwidth, which decreases the collision probability, could be observed in Fig. \ref{scc}. 
 
\begin{figure}[t!]
   \centering
 	\includegraphics[width=3.5in]{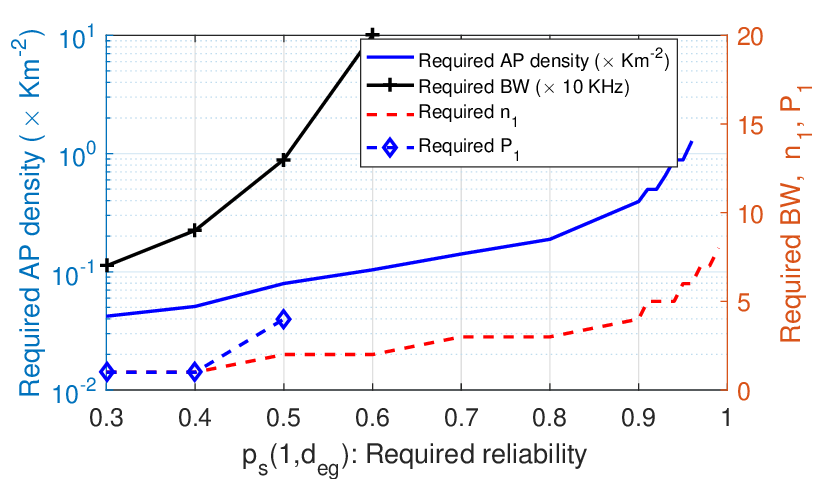}
 	\caption{Scalability analysis versus required reliability ($K=1$, $\lambda_1$=3.2).  }
 	\label{scc}
 \end{figure}

 \section{Conclusion}
A tractable analytical model of reliability in large-scale heterogeneous IoT networks has been presented as a function of IoT traffic intensity and access network's resources. This model has been employed to analyze the impacts of resource provisioning at the network-side and operation control at the device-side on reliability and battery lifetime of IoT devices. The derived expressions illustrate the rate of increase in reliability and battery lifetime achieved by increasing the bandwidth of communications and number of APs. Our analyses indicated that depending on the operating point, increasing transmit power and number of replica transmissions may increase both reliability and battery lifetime, offer a tradeoff between them, or decrease both of them. Then, we developed a lifetime-optimal operation control policy for IoT devices. The simulation results confirmed existence of such an optimal operation point before which, battery lifetime and reliability are increasing in transmit power and number of replica transmissions; while beyond that point, there is a tradeoff between them. Finally, we have presented the scalability analysis to figure out the bounds up to which, increasing the provisioned resources at the network-side, or increasing energy consumption of IoT devices per packet transfer, can compensate the impact of increase in number of  devices or their required QoS. The tightness and tractability of the derived expressions   promote use of them in IoT-network planning and operation control.


     \ifCLASSOPTIONcaptionsoff
  \newpage
\fi

\bibliographystyle{IEEEtran}
\bibliography{bibs}
 \end{document}